\documentclass[lettersize,journal]{IEEEtran}
\usepackage{amsmath,amsfonts,amssymb}
\usepackage{array}
\usepackage[caption=false,font=normalsize,labelfont=sf,textfont=sf]{subfig}
\usepackage{textcomp}
\usepackage{stfloats}
\usepackage{url}
\usepackage{verbatim}
\usepackage{graphicx}
\usepackage{cite}
\usepackage{paracol}
\usepackage{booktabs}
\usepackage{makecell}
\usepackage{tabularx}
\usepackage{multirow}
\usepackage{xcolor}
\usepackage[normalem]{ulem}  
\usepackage{algorithm}
\usepackage{algpseudocode}
\usepackage{amsthm}
\newtheorem{theorem}{Theorem}

\theoremstyle{definition}
\newtheorem{definition}[theorem]{Definition}

\theoremstyle{remark}

\begin{document}
	
\title{From Leaves to Clusters: Depth-Efficient SAT-Oracle Synthesis Based on the FORM Model}

\author{Zhihang~Li, Wei~Zi, Shuai~Yang, Bei~Zhou, Zongjiang~Yi, Woji~He, Yingjie~Lin, Kaixiang~Ji, Heru~Du, and Jinchen~Xu
	\thanks{Zhihang Li and Wei Zi contributed equally to this work.}
	\thanks{Zhihang Li, Bei Zhou, Zongjiang Yi, Woji He, Yingjie Lin, Kaixiang Ji, Heru Du, and Jinchen Xu are with the Laboratory for Advanced Computing and Intelligence Engineering, Zhengzhou 450001, China (e-mail: tiancaizhi666@163.com).}
	\thanks{Wei Zi is with the Quantum Science Center of Guangdong-Hong Kong-Macao Greater Bay Area, Shenzhen 518045, China. The work of Wei Zi was supported by the Guangdong Provincial Quantum Science Strategic Initiative under Grant GDZX2503001.}
	\thanks{Shuai Yang is with Hangzhou Dianzi University, Hangzhou, China. The work of Shuai Yang was supported by the National Natural Science Foundation of China under Grant 62402149.}
	\thanks{Corresponding author: Jinchen Xu (e-mail: atao728208@126.com).}}

\markboth{IEEE Transactions on Computer-Aided Design of Integrated Circuits and Systems,~Vol.~XX, No.~X, 2026}%
{Li \MakeLowercase{\textit{et al.}}: From Leaves to Clusters: Depth-Efficient SAT-Oracle Synthesis Based on the FORM Model}


\maketitle

\begin{abstract}
Quantum oracles play a fundamental role in many quantum algorithms, such as Grover's search and quantum counting. In oracle synthesis, reducing the circuit depth under a limited number of ancilla qubits is essential, given the constraints of current quantum devices. However, for satisfiability oracles (SAT-oracles) under a limited ancilla budget, existing synthesis methods fail to fully exploit clause-level parallelism, leading to deep circuits. In practice, evaluating clauses in parallel can quickly exhaust the ancilla budget. To address this issue, we introduce the Clustered Synthesis Tree (CST), a depth-oriented SAT-oracle synthesis framework that exploits clause-level parallelism while respecting the ancilla budget. CST leverages a recursive synthesis tree as the backbone for scheduling clause evaluations and groups its clause leaves into clusters, whose clauses are evaluated in parallel by a reversible oracle called ClausePack. To form the clusters, CST introduces a polynomial-time heuristic called SeedGrow for the NP-hard clause-grouping problem, and the resulting tree is compiled into an executable circuit by CST-Map. Compared with the depth-oriented pAND baseline, CST can reduce the circuit depth by an average of $10.2\times$ and $8.8\times$--$11.0\times$ on SATLIB and random $4$-CNF benchmarks, respectively, under the same ancilla budgets, while also outperforming the width-oriented pAND, clause-storing, and qSAT baselines.
\end{abstract}


\section{Introduction}

Quantum oracles are widely used in quantum algorithms, such as Grover's search~\cite{ref:grover-oracle,ref:deterministic-grover,ref:amplitude-amplification} and quantum counting~\cite{ref:quantum-counting,ref:approx-counting}, as well as in quantum approaches to constraint satisfaction and optimization~\cite{ref:sat-in-quantum,ref:sat-oracle-applications}. To execute these algorithms on a quantum device, it is necessary to synthesize each oracle into a circuit of elementary gates~\cite{ref:oracle-synthesis-general,ref:reversible-logic-synthesis}. Since an oracle is invoked multiple times within a single run of such an algorithm and often accounts for most of the resulting circuit~\cite{zhang2025pebbling}, oracle circuit optimization has become an important research direction. The circuit depth and the number of ancilla qubits are among the most prominent cost measures of an oracle circuit: the former corresponds to its execution time, and the latter to its space cost~\cite{ref:time-space-tradeoff,yang2024sat}.

There are several synthesis methods for quantum oracles~\cite{ref:truth-table-synthesis,ref:dd-synthesis,ref:esop-synthesis,zheng2025sshr,nie2026oracle}. These methods synthesize oracles based on different Boolean representations, such as truth tables, decision diagrams, and ESOP expressions. However, for conjunctive normal form (CNF) formulas, these methods can require exponential running time in the worst case~\cite{yang2024sat}. A CNF formula is a conjunction of clauses, each of which is a disjunction of literals. The oracle of a CNF formula is referred to as a SAT-oracle~\cite{yang2024sat}.

SAT-oracles are the basic components of Grover-based SAT solvers~\cite{ref:qiskit-aqua,ref:sat-in-quantum}. In quantum SAT solving, many constructions have emerged to improve the clause encoding of SAT-oracles. For example, Kole et al.~\cite{kole2025qsat} derived ESOP-based clauses for hardware equivalence checking, reducing the qubit count, gate count, and circuit depth. Assaad et al.~\cite{assaad2026esop} extended this ESOP-based encoding to general quantum SAT solving. Other constructions have been developed to make the ancilla count adjustable. For instance, Yang et al.~\cite{yang2024sat} proposed a width-oriented construction and a depth-oriented construction, which reduce the ancilla count and the circuit depth, respectively. However, these constructions either consume a large number of ancilla qubits or, when ancilla qubits are limited, suffer from deep circuits due to the repeated evaluation of sub-blocks. Lin et al.~\cite{lin2024parallel} duplicate shared variables to evaluate all clauses in parallel, but still require a number of qubits that grows with the formula size, regardless of the ancilla budget. Omanakuttan et al.~\cite{omanakuttan2025threshold} employ graph coloring to evaluate clauses without shared variables in parallel when estimating the resources of QAOA-based $k$-SAT solving on fault-tolerant quantum computers, but still evaluate clauses that share variables serially in separate batches. Overall, existing SAT-oracle constructions still do not simultaneously take into account both the variables shared among clauses and the ancilla qubits required to duplicate them when scheduling clause evaluations under a given ancilla budget.

To address the above limitation, this work analyzes the clause-level parallelism of SAT-oracles and proposes the Clustered Synthesis Tree (CST), a depth-oriented framework to synthesize SAT-oracles under a limited ancilla budget. Specifically, CST first adopts the synthesis tree of our previous work, constructed by the FORM model and the ShallowGrow algorithm~\cite{ref:hrse-asdt}, as the backbone and groups its clause leaves into multiple clusters, each evaluated as a single leaf. Especially for the clauses in a cluster, it incorporates a reversible oracle called ClausePack, which evaluates these clauses in parallel by duplicating their shared variables, thereby reducing the circuit depth. Second, by analyzing the variables shared among clauses, CST introduces a polynomial-time heuristic called SeedGrow for the NP-hard clause-grouping problem, enabling the clauses to be grouped into fewer clusters within the ancilla budget. Finally, CST incorporates a mapping procedure called CST-Map to compile the resulting tree into an executable SAT-oracle circuit. By selecting appropriate clusters for each node of the tree, CST achieves clause-level parallelism within the ancilla budget while preserving the minimal repeated evaluation of the backbone, thereby reducing the overall circuit depth. Compared to the depth-oriented pAND baseline~\cite{yang2024sat}, CST delivers average depth reductions of $10.2\times$ and $8.8\times$--$11.0\times$ on SATLIB and random $4$-CNF benchmarks, respectively, under the same ancilla budgets.

The main contributions of this article are as follows.
\begin{enumerate}
\item We propose CST, an end-to-end SAT-oracle synthesis framework under a limited ancilla budget, which jointly considers the variables shared among clauses and the ancilla qubits required to duplicate them, and compiles the resulting synthesis tree into an executable circuit through CST-Map.
\item We propose a reversible oracle, ClausePack, in CST, which enables it to evaluate the $t$ clauses of a cluster in parallel with a depth of $O(D_{\max}+\log t)$, where $D_{\max}$ is the maximum depth of a single clause.
\item We develop a polynomial-time heuristic, SeedGrow, for the clause-grouping problem in CST, which we prove NP-complete on $k$-CNF instances for every fixed $k\ge3$. SeedGrow groups the clauses into fewer clusters within the ancilla budget, so that more clauses are evaluated in parallel by ClausePack.
\item We evaluate CST on SATLIB and random $4$-CNF benchmarks. Compared to the depth-oriented pAND baseline under the same ancilla budgets, CST achieves average depth reductions of $10.2\times$ and $8.8\times$--$11.0\times$ on the two benchmarks, respectively, and also outperforms the width-oriented pAND, clause-storing, and qSAT baselines. In the resource estimation for Grover search, the full Grover circuit built on CST is $64\%$--$75\%$ shallower than that built on the depth-oriented pAND oracle.
\end{enumerate}

The remainder of this article is organized as follows. Section~\ref{sec:preliminaries} introduces the problem formulation and reviews the FORM/ShallowGrow backbone. Section~\ref{sec:cst} introduces the CST model and the overall synthesis pipeline built on it. Section~\ref{sec:scheduling} formalizes the clause-grouping problem, proves its NP-completeness, and describes SeedGrow. Section~\ref{sec:method} presents the ClausePack oracle. Section~\ref{sec:mapping} describes the CST-Map procedure. Section~\ref{sec:experiments} reports the experimental results. Finally, Section~\ref{sec:conclusion} concludes this article.

\section{Problem Formulation and Preliminaries}
\label{sec:preliminaries}

\subsection{Problem Formulation}
\label{subsec:problem}

Let \(x=(x_1,\ldots,x_n)\) denote \(n\) Boolean variables.
A \(k\)-CNF formula with \(m\) clauses is written as
\begin{equation}
f(x)=C_1(x)\land C_2(x)\land\cdots\land C_m(x),
\end{equation}
where each clause \(C_i\) is a disjunction of exactly \(k\) literals, and each literal is either a variable or its negation.
We denote such an instance by \(\mathrm{CNF}_{n,m}^{k}\).

For such a formula \(f\), the corresponding SAT-oracle is the reversible operator
\begin{equation}
\label{eq:sat-oracle}
U_f:\;|x\rangle|c\rangle\;\mapsto\;|x\rangle|c\oplus f(x)\rangle,
\end{equation}
acting on the $n$-qubit input register $|x\rangle$ and a target qubit $|c\rangle$. The target qubit is toggled if and only if the assignment $x$ satisfies $f$, i.e., $f(x)=1$.

In a standard reversible implementation, each clause predicate is first evaluated into a work qubit, and the resulting clause values are then combined on the target qubit by a multi-controlled NOT (MCX) gate~\cite{ref:basic-sat-oracle}.

\begin{table}[t]
\centering
\caption{Main notation.}
\label{tab:notation}
\begin{tabular}{cl}
\toprule
Symbol & Meaning \\
\midrule
$n$        & number of Boolean variables \\
$m$        & number of clauses \\
$k$        & clause width \\
$a_q$      & ancilla-qubit budget \\
$D(\cdot)$ & circuit depth \\
$R(\cdot)$ & variable-replication redundancy of a clause cluster \\
$s(v)$     & node size (ancillae available to a unit) \\
$d(v)$     & node depth in the synthesis tree \\
$\kappa(v)$& node out-degree \\
$c(v)$     & node complexity in the FORM tree \\
$\ell(v)$  & covered leaf count of the subtree at $v$ \\
\bottomrule
\end{tabular}
\end{table}

\subsection{The FORM/ShallowGrow Backbone}
\label{subsec:backbone}

FORM~\cite{ref:hrse-asdt} formalizes recursive oracle structures as a model for complexity analysis. ShallowGrow uses this model to construct a synthesis backbone that minimizes repeated function evaluations under a fixed ancilla budget.

\begin{definition}[FORM model]
A FORM model represents an oracle by a rooted computation tree \(T=(V,E,A)\). The vertex set \(V\) contains computation units, the edge set \(E\) encodes sub-computation relations, and \(A\) assigns each node \(v\) the attribute tuple
\begin{equation}
A(v)=\bigl(s(v),d(v),\kappa(v),c(v),\ell(v)\bigr).
\end{equation}
Here \(s(v)\) is the available ancilla size, \(d(v)\) is the tree depth, \(\kappa(v)\) is the out-degree, \(c(v)\) is the node complexity, and \(\ell(v)\) is the number of covered leaves. Table~\ref{tab:notation} summarizes the main notation. Leaves represent individual clause evaluations, while internal nodes recursively combine the values produced by their children. A valid FORM tree satisfies the following constraints. The node-size monotonicity constraint ensures that a child cannot require more ancillae than its parent:
\begin{equation}
\label{eq:hrse-size-monotonicity}
0 \le s(v_j) \le s(v_i), \qquad (v_i,v_j)\in E.
\end{equation}
The sibling-size distinction constraint requires the children $v_{j_1},v_{j_2},\ldots$ of any common parent to have pairwise distinct sizes:
\begin{equation}
\label{eq:hrse-sibling-size-distinction}
s(v_{j_p}) \neq s(v_{j_q}), \qquad p\neq q.
\end{equation}
The terminal-node constraint prevents nodes of size \(1\) or \(2\) from being further expanded:
\begin{equation}
\label{eq:hrse-terminal-node}
s(v_i)=1 \;\vee\; s(v_i)=2 \quad \Longrightarrow \quad \kappa(v_i)=0 .
\end{equation}
\end{definition}

Using this tree representation, the total oracle complexity is given by the root-node complexity:
\begin{equation}
c(v_0)=
\sum_{\substack{v\in V\\ \kappa(v)=0}}2^{d(v)}\tau+
\sum_{\substack{u\in V\\ \kappa(u)\neq0}}2^{d(u)}\Gamma(u),
\end{equation}
where \(\tau\) is the cost of a single leaf-function evaluation and \(\Gamma(u)\) is the merging cost at an internal node. Accordingly, the number of repeated function evaluations induced by \(T\) is
\begin{equation}
N_{\mathrm{eval}}(T)=
\sum_{\substack{v\in V\\ \kappa(v)=0}}2^{d(v)}.
\end{equation}

Given \(m\) clause functions and an ancilla budget \(a_q\), ShallowGrow grows a FORM tree from a root of size \(a_q+1\) until \(m\) leaves are obtained; the extra unit is the oracle target qubit, which serves as the root's output slot and is not charged to \(a_q\), whereas every other node stores its output on one of the qubits it receives from its parent. ShallowGrow expands a node of size \(s\) into children of the distinct sizes \(s-1,\ldots,1\), and nodes of size \(1\) or \(2\) are leaves, so a node of size \(s\ge2\) covers at most \(2^{s-2}\) leaves; hence a budget \(a_q\) accommodates at most \(2^{a_q-1}\) clauses, and the tree exists if and only if \(a_q\ge\lceil\log_2 m\rceil+1\). It selects expansion candidates using the minimum-depth strategy and the maximum-size strategy; under the FORM cost model, this construction minimizes the repeated function-evaluation count \(N_{\mathrm{eval}}(T)\)~\cite{ref:hrse-asdt}. Since ShallowGrow treats each leaf as an abstract clause function, the resulting backbone does not exploit clause-level sharing. To account for such clause-level structure, the next section introduces the CST model, which refines the ShallowGrow backbone by replacing singleton clause leaves with feasible clustered leaves.

\section{The Clustered Synthesis Tree Model}
\label{sec:cst}

The CST model represents a SAT-oracle synthesis plan by replacing singleton clause leaf nodes with clustered leaf nodes inside a recursive oracle tree.

\subsection{From Clause Leaves to Clause Clusters}
\label{subsec:cst-model}

Consider a FORM/ShallowGrow tree \(T=(V,E,A)\) whose leaves correspond to the clause functions \(C_1,\ldots,C_m\). In the unclustered tree, each leaf represents one clause value, so clauses that could be evaluated in parallel remain separated at the structural level.

The CST model refines only the leaf layer of this recursive tree. For an internal computation node \(v\), the clause leaf nodes attached directly to \(v\) as singleton clause leaves are grouped into an ordered partition of clause clusters,
\begin{equation}
\Pi(v)=(\mathcal{B}_{v,1},\mathcal{B}_{v,2},\ldots,\mathcal{B}_{v,K_v}).
\end{equation}
Each cluster \(\mathcal{B}_{v,j}\) becomes one clustered leaf and is evaluated by one invocation of a parallel clause-evaluation oracle. The recursive backbone still determines the nesting of sub-computations, while \(\Pi(v)\) determines which clause leaf nodes are merged into each clustered leaf.

\subsection{Feasibility of Clause Clusters}
\label{subsec:cst-feasibility}

A clause cluster is feasible only when its parallel evaluation can be realized within the available ancillae. For a cluster \(\mathcal{B}\), the redundancy \(R(\mathcal{B})\) denotes the number of additional variable copies required by shared variables. For a fixed internal computation node, its clustered leaves are evaluated sequentially, and the same ancilla region is reused across these clustered leaves. The ancillae available for this evaluation are the node-local budget \(a_q(v)\) assigned by the ShallowGrow backbone, with \(a_q(v_0)=a_q\) at the root \(v_0\) and \(a_q(v)\le a_q\) at deeper nodes. Therefore, a partition \(\Pi(v)\) is feasible only if
\begin{equation}
\label{eq:cst-prefix-feasibility}
\sum_{h=1}^{j} |\mathcal{B}_{v,h}| + R(\mathcal{B}_{v,j}) \le a_q(v),
\qquad j=1,\ldots,K_v .
\end{equation}
The first term accounts for the clause-value qubits already occupied by earlier clustered leaves of \(v\), while \(R(\mathcal{B}_{v,j})\) accounts for the variable-replication space required by the current clustered leaf. Because each node clusters only its own direct clause leaves rather than all \(m\) clauses in its subtree, no single partition must accommodate all \(m\) clauses; the ancilla-scarce regime \(a_q<m\) is thus absorbed by the recursive tree structure inherited from ShallowGrow.

\subsection{Formal Definition of CST}
\label{subsec:cst-definition}

\begin{definition}[CST model]
Given a FORM/ShallowGrow backbone \(T=(V,E,A)\) for a CNF formula \(f=C_1\land\cdots\land C_m\), a CST is a pair \(\mathcal{T}=(T,\mathcal{P})\), where \(\mathcal{P}\) assigns to each internal computation node \(v\in V\) an ordered cluster partition \(\Pi(v)\) of the clause leaf nodes attached directly to \(v\). Each element \(\mathcal{B}_{v,j}\in\Pi(v)\) is a clustered leaf that represents a set of clauses evaluated by a single parallel clause-evaluation oracle. A CST is feasible if every assigned partition satisfies the prefix feasibility constraint in \eqref{eq:cst-prefix-feasibility}.
\end{definition}

The unclustered ShallowGrow backbone is the special case in which every clustered leaf contains exactly one clause leaf.

\subsection{Synthesis Pipeline Based on CST}
\label{subsec:framework-overview}

The CST model separates SAT-oracle synthesis into four stages, as illustrated in Fig.~\ref{fig:cst-pipeline}. ShallowGrow first supplies the recursive backbone under the ancilla budget \(a_q\). Clause grouping then constructs feasible ordered cluster partitions. Each clustered leaf is realized by a parallel clause-evaluation oracle, and the resulting feasible CST is finally compiled into an executable SAT-oracle circuit by CST-Map.

\begin{figure}[t]
\centering
\includegraphics[width=\columnwidth]{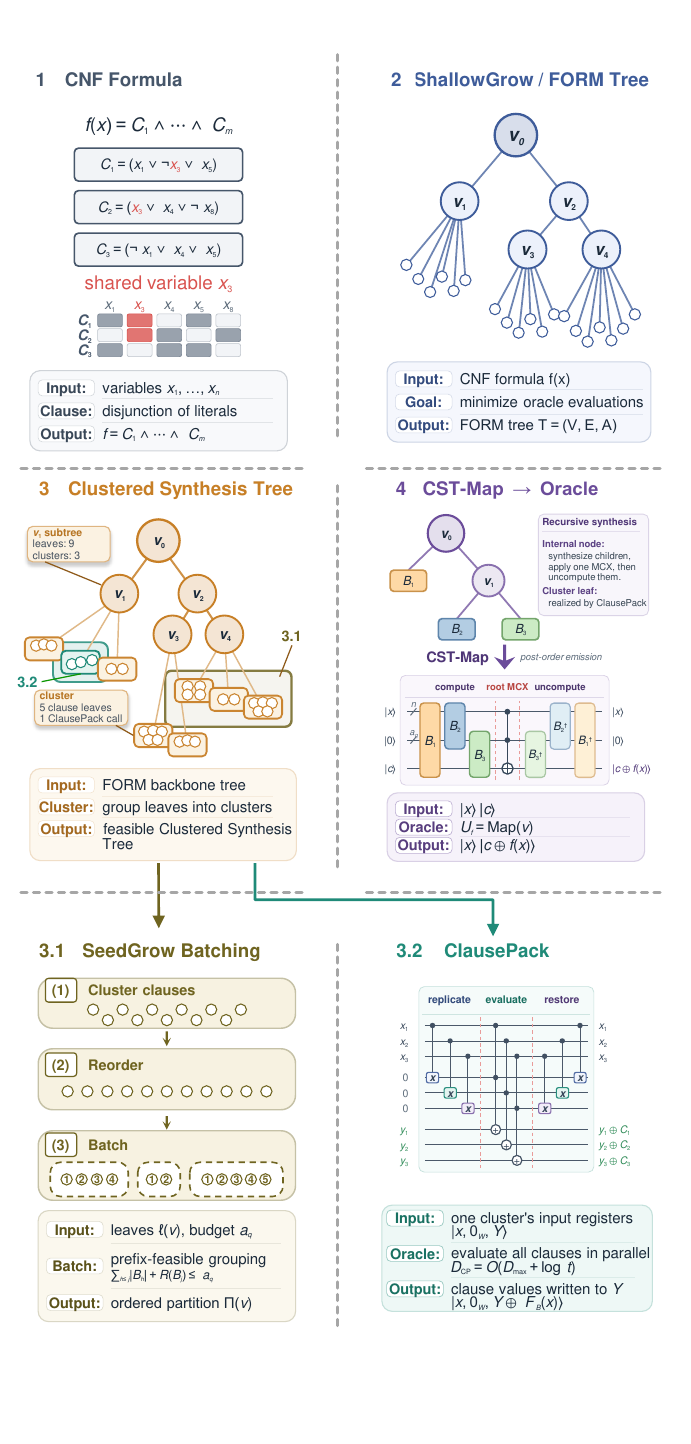}
\caption{The CST-based SAT-oracle synthesis pipeline. A CNF instance is processed through four stages---ShallowGrow backbone construction under the ancilla budget \(a_q\), clause clustering into feasible ordered partitions, per-leaf ClausePack evaluation, and CST-Map compilation---yielding an executable SAT-oracle.}
\label{fig:cst-pipeline}
\end{figure}

\section{Clause Grouping}
\label{sec:scheduling}

Given the CST model in Section~\ref{sec:cst}, clause grouping determines the ordered cluster partition \(\Pi(v)\) assigned to each internal computation node \(v\). The input is the set of clause leaf nodes attached directly to \(v\), and the output is a sequence of clustered leaves that can be evaluated under the ancilla budget.

For each clause \(C_i\), let \(\widehat{C}_i\) be the set of variables appearing in \(C_i\), ignoring whether each variable appears positively or negatively. For a cluster \(\mathcal{B}\), let \(\nu_z(\mathcal{B})\) be the number of clauses in \(\mathcal{B}\) that contain variable \(z\). The variable-replication redundancy of the cluster is
\begin{equation}
\label{eq:cluster-redundancy}
R(\mathcal{B})=
\sum_z \max\{\nu_z(\mathcal{B})-1,0\}.
\end{equation}
This quantity counts the additional copies required to supply shared variables to parallel clause-evaluation modules. At a fixed CST computation node, clustered leaves are evaluated one after another; hence the cluster count \(K\) equals the number of serial parallel-evaluation stages and directly controls the node-level depth. The grouping problem is therefore to construct an ordered partition \((\mathcal{B}_1,\ldots,\mathcal{B}_K)\) that minimizes \(K\) while satisfying the CST prefix feasibility constraint.

\subsection{Scheduling Model and Computational Complexity}
\label{subsec:scheduling-complexity}

To characterize the computational difficulty of clause grouping, consider one fixed internal computation node of the CST. Its direct clause leaves are represented by the variable sets \(\widehat{C}_1,\widehat{C}_2,\ldots,\widehat{C}_m\) defined above; with slight abuse of notation, \(m\) here denotes the number of direct clause leaves at this node, which coincides with the global clause count only at the root, and the budget \(a_q\) denotes the node-local budget \(a_q(v)\), equal to the global budget only at the root. The node-level scheduling question is whether these clauses can be arranged into at most \(L\) feasible clustered leaves under the ancilla budget \(a_q\).

\begin{definition}[Ancilla-Constrained Clause Scheduling]
	\label{def:accs-decision}
	Given variable sets \(\widehat{C}_1,\widehat{C}_2,\ldots,\widehat{C}_m\), an ancilla budget \(a_q\), and an integer \(L\), determine whether there exists an ordered partition of \(\{C_1,\ldots,C_m\}\) into nonempty clusters,
	\begin{equation}
	\Pi = (\mathcal{B}_1,\mathcal{B}_2,\ldots,\mathcal{B}_K),
	\qquad K \le L,
	\end{equation}
	such that
	\begin{equation}
	\label{eq:accs-prefix-feasibility}
	\sum_{h=1}^{j} |\mathcal{B}_h| + R(\mathcal{B}_j) \le a_q,
	\qquad \forall j=1,\ldots,K.
	\end{equation}
	This problem is denoted by \textsc{ACCS-D}.
\end{definition}

The corresponding optimization problem minimizes the number of clustered leaves \(K\) subject to the same feasibility constraints.

\begin{theorem}
	\label{thm:accs-npc}
	For every fixed integer \(k\ge3\), \textsc{ACCS-D} is NP-complete on \(\mathrm{CNF}_{n,m}^{k}\) instances, even when \(L=2\) and \(a_q=m\). Consequently, the optimization problem of minimizing the number of clusters under ancilla constraints is NP-hard on \(k\)-CNF instances.
\end{theorem}

\begin{proof}
	Membership in NP follows from direct verification. A certificate is an ordered partition
	\(
	\Pi = (\mathcal{B}_1,\ldots,\mathcal{B}_K)
	\)
	with \(K\le L\); since the clusters are nonempty and disjoint, \(K\le m\), so the certificate has polynomial size. A verifier checks in polynomial time that the clusters form a partition of \(\{C_1,\ldots,C_m\}\), computes each redundancy value \(R(\mathcal{B}_j)\), and tests the \(K\) prefix inequalities in~\eqref{eq:accs-prefix-feasibility}. Thus \textsc{ACCS-D} on \(\mathrm{CNF}_{n,m}^{k}\) instances belongs to NP.

	NP-hardness follows from a polynomial reduction from the independent-set decision problem on \(3\)-regular graphs, hereafter denoted \textsc{Cubic-IS}: given a \(3\)-regular simple graph \(G=(V_G,E_G)\) and a threshold \(\theta{}\in\mathbb{Z}\), decide whether \(G\) admits an independent set of size at least \(\theta{}\). This problem is NP-complete~\cite{ref:three-coloring-npc}. Let \((G,\theta{})\) be a \textsc{Cubic-IS} instance with \(N:=|V_G|\ge4\); the handshake lemma implies that \(N\) is even, so \(3N/2\in\mathbb{Z}\).

	\emph{Threshold normalization.} If \(\theta{}>3N/8\), set \(\lambda{}:=\lceil(8\theta{}-3N)/4\rceil\ge1\) and let \(G'\) be the vertex-disjoint union of \(G\) with \(\lambda{}\) copies of the complete graph \(K_4\). Every vertex of a \(K_4\) has degree three within its own copy, so \(G'\) remains \(3\)-regular with \(N'=N+4\lambda{}\) (even) and \(\alpha(G')=\alpha(G)+\lambda{}\); set the updated threshold \(\theta{}':=\theta{}+\lambda{}\). By the choice of \(\lambda{}\), \(4\lambda{}\ge8\theta{}-3N\), which rearranges to \(\theta{}'\le3N'/8\), while \(\alpha(G)\ge \theta{}\) if and only if \(\alpha(G')\ge \theta{}'\). Rename \((G',\theta{}')\) to \((G,\theta{})\); henceforth assume \(\theta{}\le3N/8\).

	\emph{Main construction.} Set \(M:=\tfrac{3N}{2}-4\theta{}\ge0\); \(M\) is an integer because \(N\) is even. Construct the following \textsc{ACCS-D} instance:
	\begin{itemize}
		\item For each vertex \(u_i\in V_G\), introduce a \emph{vertex clause} \(C_i\) whose variable set contains one fresh \emph{edge variable} \(x_e\) for each edge \(e\in E_G\) incident to \(u_i\), together with \(k-3\) fresh \emph{private variables} \(y_{i,1},\ldots,y_{i,k-3}\); hence \(|\widehat{C}_i|=3+(k-3)=k\). Each edge variable \(x_e\), with \(e=\{u_i,u_j\}\in E_G\), appears in exactly the two vertex clauses \(C_i\) and \(C_j\), and in no other clause.
		\item For \(j=1,\ldots,M\), introduce a \emph{filler clause} \(P_j\) whose variable set consists of \(k\) fresh private variables \(z_{j,1},\ldots,z_{j,k}\); hence \(|\widehat{P}_j|=k\).
		\item The private variables \(y_{i,\ell}\) and \(z_{j,\ell}\) are all distinct from one another and from every edge variable, and each occurs in exactly one clause.
	\end{itemize}
	Set the cluster bound \(L=2\) and the ancilla budget \(a_q=m=N+M\). Every \(\widehat{C}_i\) and every \(\widehat{P}_j\) has size exactly \(k\), so the resulting instance belongs to \(\mathrm{CNF}_{n,m}^{k}\) in the sense of Section~\ref{subsec:problem}. The instance has size \(O(kN)\), and the reduction runs in polynomial time.

	\emph{Redundancy identity.} For any cluster \(\mathcal{B}\), let \(V_\mathcal{B}:=\{u_i:C_i\in\mathcal{B}\}\subseteq V_G\). Every private variable occurs in a single clause and contributes 0 to \(R(\mathcal{B})\). Every edge variable \(x_e\) with \(e=\{u,v\}\in E_G\) satisfies \(\nu_{x_e}(\mathcal{B})=|\{u,v\}\cap V_\mathcal{B}|\), and thus contributes 1 to \(R(\mathcal{B})\) exactly when both endpoints of \(e\) lie in \(V_\mathcal{B}\), and 0 otherwise. Writing \(e_G(V_\mathcal{B})\) for the number of edges of \(G\) with both endpoints in \(V_\mathcal{B}\),
	\begin{equation}
	\label{eq:accs-redundancy-identity}
	R(\mathcal{B})=e_G(V_\mathcal{B}).
	\end{equation}

	\emph{Correctness (if).} Assume \(\alpha(G)\ge \theta{}\), and fix an independent set \(I^\ast\subseteq V_G\) with \(|I^\ast|=\theta{}\). Define
	\begin{equation*}
	\mathcal{B}_2:=\{C_i:u_i\in I^\ast\}\cup\{D_1,\ldots,P_M\},\quad \mathcal{B}_1:=\{C_i:u_i\in V_G\setminus I^\ast\}.
	\end{equation*}
	Since \(\theta{}\le3N/8<N\), \(V_G\setminus I^\ast\) is nonempty, and hence so is \(\mathcal{B}_1\); \(\mathcal{B}_2\) is nonempty because either it contains a vertex clause (when \(\theta{}\ge1\)) or it contains all \(M=\tfrac{3N}{2}\ge6\) filler clauses (when \(\theta{}=0\)). By~\eqref{eq:accs-redundancy-identity}, \(R(\mathcal{B}_2)=e_G(I^\ast)=0\) and \(R(\mathcal{B}_1)=e_G(V_G\setminus I^\ast)\). Because \(G\) is \(3\)-regular and \(I^\ast\) is independent, every edge incident to \(I^\ast\) has its other endpoint in \(V_G\setminus I^\ast\), so \(|E(I^\ast,V_G\setminus I^\ast)|=3\theta{}\) and \(e_G(V_G\setminus I^\ast)=|E_G|-3\theta{}=\tfrac{3N}{2}-3\theta{}\). The two prefix inequalities in~\eqref{eq:accs-prefix-feasibility} become
	\begin{equation*}
	|\mathcal{B}_1|+R(\mathcal{B}_1)=(N-\theta{})+\left(\tfrac{3N}{2}-3\theta{}\right)=\tfrac{5N}{2}-4\theta{}=N+M=a_q,
	\end{equation*}
	\begin{equation*}
	|\mathcal{B}_1|+|\mathcal{B}_2|+R(\mathcal{B}_2)=m+0=a_q,
	\end{equation*}
	both satisfied.

	\emph{Correctness (only if).} Conversely, assume the constructed instance admits a feasible ordered partition \(\Pi=(\mathcal{B}_1,\ldots,\mathcal{B}_K)\) with \(K\le2\). If \(K=1\), then \(\mathcal{B}_1\) contains all clauses and \(R(\mathcal{B}_1)=e_G(V_G)=\tfrac{3N}{2}>0\), so the prefix inequality \(m+R(\mathcal{B}_1)\le a_q=m\) fails, a contradiction. Hence \(K=2\). Let \(V_2:=\{u_i:C_i\in\mathcal{B}_2\}\) and \(P^{(2)}:=\{j:P_j\in\mathcal{B}_2\}\); write \(V_1:=V_G\setminus V_2\). The prefix inequality at \(j=2\) gives \(m+R(\mathcal{B}_2)\le m\), so by~\eqref{eq:accs-redundancy-identity} \(e_G(V_2)=0\), i.e., \(V_2\) is an independent set of \(G\). The prefix inequality at \(j=1\) gives
	\begin{equation*}
	(N-|V_2|)+(M-|P^{(2)}|)+e_G(V_1)\le N+M,
	\end{equation*}
	that is, \(e_G(V_1)\le|V_2|+|P^{(2)}|\le|V_2|+M\). Since \(V_2\) is independent in the \(3\)-regular graph \(G\), \(e_G(V_1)=\tfrac{3N}{2}-3|V_2|\); substituting,
	\begin{equation*}
	\tfrac{3N}{2}-3|V_2|\le|V_2|+M=|V_2|+\tfrac{3N}{2}-4\theta{},
	\end{equation*}
	which simplifies to \(|V_2|\ge \theta{}\). Thus \(V_2\) is an independent set of \(G\) of size at least \(\theta{}\), so \(\alpha(G)\ge \theta{}\).

	Combined with membership in NP, \textsc{ACCS-D} on \(\mathrm{CNF}_{n,m}^{k}\) instances is NP-complete for every fixed \(k\ge3\). Because the reduction uses only polynomially bounded numeric parameters (\(a_q=m=O(N)\) and constant \(L=2\)), the result holds in the strong sense, and therefore no pseudo-polynomial-time algorithm exists unless \(\mathrm{P}=\mathrm{NP}\). Finally, the decision problem reduces to the optimization problem by comparing the minimum feasible number of clusters with \(L\); hence minimizing the number of clusters under ancilla constraints on \(k\)-CNF instances is NP-hard.
\end{proof}

Exact minimization of the cluster count is therefore impractical for large CST instances, which motivates the polynomial-time heuristic below.

\subsection{The Seed-Growing Batch Heuristic}
\label{subsec:seedgrow}

SeedGrow is a polynomial-time heuristic that packs the direct clause leaves of a CST node into as few clusters as possible, each called a \emph{batch} in Algorithms~\ref{alg:bgp}--\ref{alg:growblock}, and replication redundancy enters only through the prefix feasibility constraint~\eqref{eq:cst-prefix-feasibility}. The heuristic is motivated by the observation that the available redundancy budget depends on the cluster position: later clusters have less remaining ancilla space because more clause-value qubits have already been occupied.

SeedGrow constructs the partition in reverse evaluation order. The last cluster is generated first because it has the smallest redundancy allowance, \(a_q-m\). After a cluster \(\mathcal{B}\) has been fixed, the algorithm moves one position earlier in the evaluation order; at that position, the clause-value qubits of \(\mathcal{B}\) have not yet been occupied, so the redundancy allowance increases by \(|\mathcal{B}|\). Within each cluster, SeedGrow starts from a low-conflict seed clause. Let \(\nu_z\) denote the number of clauses in the current CST node whose normalized variable set contains variable \(z\). For a clause \(C_i\), the conflict degree
\begin{equation}
\label{eq:seedgrow-conflict-degree}
d_i=\sum_{z\in \widehat{C}_i}(\nu_z-1)
\end{equation}
measures how strongly \(C_i\) shares variables with the other clauses at this node. Starting from the clause with the smallest \(d_i\), SeedGrow repeatedly adds the candidate that incurs the smallest redundancy increment \(\delta_i=|\widehat{C}_i\cap U|\) without exceeding the current allowance, where \(U\) is the variable set already covered by the cluster. After the backward pass, the clusters are reversed into the evaluation order and \textsc{MergeAdjacent} greedily merges consecutive clusters whenever the merged cluster still satisfies the prefix feasibility constraint, so the returned partition is feasible by construction. Algorithm~\ref{alg:bgp} summarizes the complete procedure, and Algorithm~\ref{alg:growblock} gives the seed-growing subroutine used to construct one cluster. Fig.~\ref{fig:seedgrow-example} traces the execution on a concrete six-clause instance.

Let \(k=\max_i|\widehat{C}_i|\) be the maximum normalized clause width at the current CST node. With bucketed candidate queues and incremental updates of \(\delta_i\), one call to \textsc{GrowBlock} takes \(O(mk)\) time, and the backward construction invokes it at most \(m\) times because each call removes at least one clause; the normalization, the conflict degrees, and the adjacent-merge pass are bounded by the same order. Thus, SeedGrow runs in \(O(m^2k)\) time and uses \(O(mk)\) space.

\begin{algorithm}[t]
	\caption{SeedGrow: backward batch construction}
	\label{alg:bgp}
	\begin{algorithmic}[1]
		\Require ancilla budget $a_q$; clauses $\mathcal{L}=\{C_1,\dots,C_m\}$
		\Ensure ordered partition $\Pi$, or \textbf{None}
		\State \textbf{if} $a_q<m$ \textbf{then return} \textbf{None}
		\State $\mathcal{S}\gets\textsc{Normalize}(\mathcal{L})$ \Comment{variable sets, duplicate variables removed}
		\State $d\gets\textsc{ConflictDegree}(\mathcal{S})$
		\State $\mathcal{R}\gets\{1,\dots,m\}$;\ \ $\Pi_{\mathrm{rev}}\gets()$;\ \ $b\gets a_q-m$
		\While{$\mathcal{R}\neq\emptyset$}
			\State $B\gets\textsc{GrowBlock}(\mathcal{R},b,\mathcal{S},d)$ \Comment{Algorithm~\ref{alg:growblock}}
			\State $\textsc{Append}(\Pi_{\mathrm{rev}},B)$;\ \ $\mathcal{R}\gets \mathcal{R}\setminus B$;\ \ $b\gets b+|B|$
		\EndWhile
		\State $\Pi_{\mathrm{idx}}\gets\textsc{MergeAdjacent}\big(\textsc{Reverse}(\Pi_{\mathrm{rev}}),\,a_q,\,\mathcal{S}\big)$
		\State \Return $\textsc{Recover}(\Pi_{\mathrm{idx}},\mathcal{L})$ \Comment{feasible by construction}
	\end{algorithmic}
\end{algorithm}

\begin{algorithm}[t]
	\caption{\textsc{GrowBlock}: seed-growing construction of one batch}
	\label{alg:growblock}
	\begin{algorithmic}[1]
		\Require unassigned clauses $\mathcal{R}$; redundancy budget $b$; variable sets $\mathcal{S}$; conflict degree $d$
		\Ensure one batch $B$
		\State $s\gets\arg\min_{i\in \mathcal{R}}\,(d_i,\,|\widehat{C}_i|,\,i)$ \Comment{seed: least-conflicting clause}
		\State $B\gets\{s\}$;\ \ $U\gets\widehat{C}_s$;\ \ $r\gets 0$
		\For{each $i\in \mathcal{R}\setminus\{s\}$}
			\State $\delta_i\gets|\widehat{C}_i\cap U|$ \Comment{redundancy increment}
		\EndFor
		\While{$\mathcal{R}\setminus B\neq\emptyset$}
			\State $\mathcal{C}\gets\{i\in \mathcal{R}\setminus B:\delta_i\le b-r\}$
			\If{$\mathcal{C}=\emptyset$}
				\State \textbf{break}
			\EndIf
			\State $c\gets\arg\min_{i\in\mathcal{C}}(\delta_i,\,d_i,\,|\widehat{C}_i|,\,i)$
			\State $U_{\mathrm{new}}\gets\widehat{C}_c\setminus U$
			\State $B\gets B\cup\{c\}$;\ \ $r\gets r+\delta_c$;\ \ $U\gets U\cup\widehat{C}_c$
			\ForAll{$i\in \mathcal{R}\setminus B$ sharing a variable in $U_{\mathrm{new}}$}
				\State $\delta_i\gets\delta_i+|\widehat{C}_i\cap U_{\mathrm{new}}|$ \Comment{incremental update}
			\EndFor
		\EndWhile
		\State \Return $B$
	\end{algorithmic}
\end{algorithm}

\begin{figure*}[t]
  \centering
  \includegraphics[width=\textwidth]{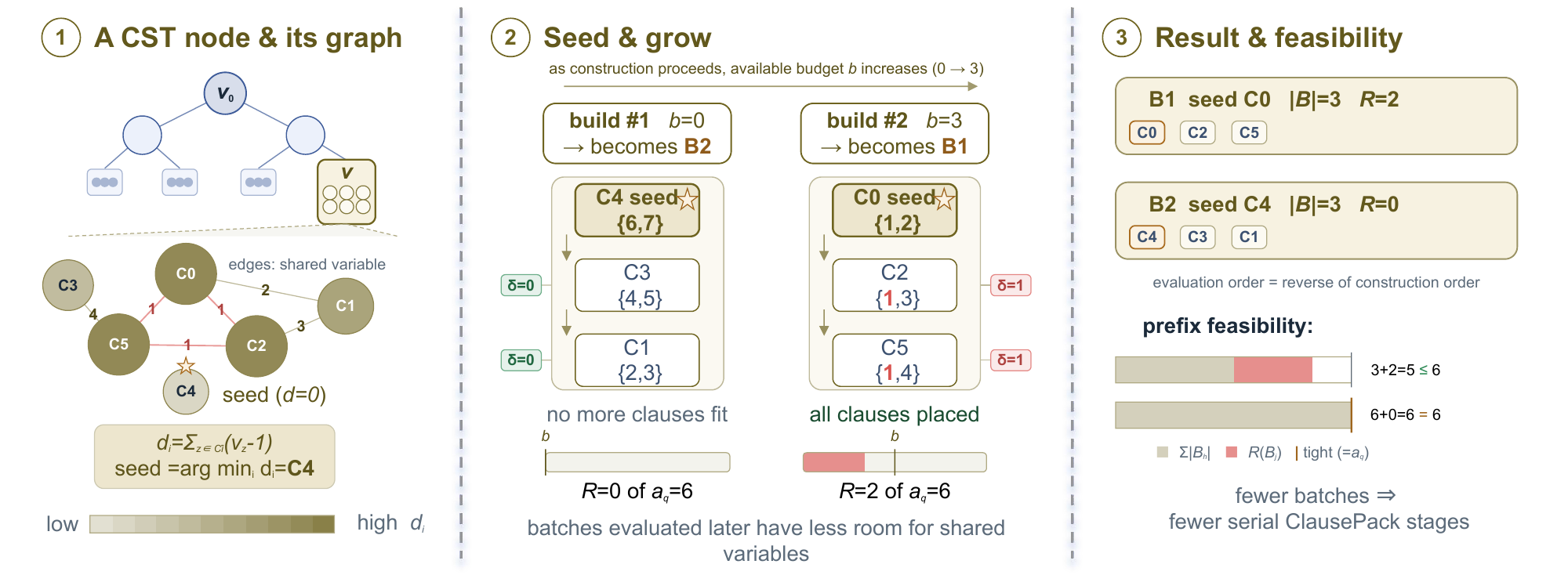}
  \caption{Worked example of \textsc{SeedGrow} on
    $f=(x_1\!\lor\!x_2)\land(x_2\!\lor\!x_3)\land(x_1\!\lor\!x_3)\land
       (x_4\!\lor\!x_5)\land(x_6\!\lor\!x_7)\land(x_1\!\lor\!x_4)$
    ($m=6$, $a_q=6$; clauses indexed from $C_0$).
    \textbf{Left:}~conflict graph shaded by $d_i$~\eqref{eq:seedgrow-conflict-degree}, with the seed~($\star$).
    \textbf{Center:}~backward construction, admitting the candidate of smallest $\delta_i$ within budget~$b$.
    \textbf{Right:}~final partition $\mathcal{B}_1=\{C_0,C_2,C_5\}$ ($R=2$) and $\mathcal{B}_2=\{C_1,C_3,C_4\}$ ($R=0$).}
  \label{fig:seedgrow-example}
\end{figure*}

\section{Parallel Clause-Evaluation Oracle}
\label{sec:method}

Section~\ref{sec:scheduling} determines which clauses can form feasible clustered leaves under the CST prefix constraint. This section specifies the circuit primitive that realizes one such clustered leaf. Variable replication is a known device for parallelizing clause evaluation: both the p-AND synthesis~\cite{yang2024sat} and Lin et al.~\cite{lin2024parallel} duplicate shared variables so that clauses sharing a variable can be evaluated concurrently. In these works, however, replication is embedded in a specific circuit construction rather than abstracted into a reusable primitive with a defined interface and resource cost, and the oracle of Lin et al.~\cite{lin2024parallel} is moreover not clean, leaving residual entangled copies at its output. ClausePack abstracts clause-cluster evaluation into a single clean reversible oracle: it evaluates all clauses in the cluster in parallel, exposes their clause-value outputs, and restores the temporary qubits used for variable replication. Its specification~\eqref{eq:clausepack-spec} and explicit depth and ancilla bounds~\eqref{eq:clausepack-depth-bound}, \eqref{eq:clausepack-space-cost} let each clustered leaf of the CST be scheduled within the ancilla budget. Fig.~\ref{fig:clausepack} illustrates this three-stage structure on an example cluster.

\begin{figure}[t]
	\centering
	\includegraphics[width=\columnwidth]{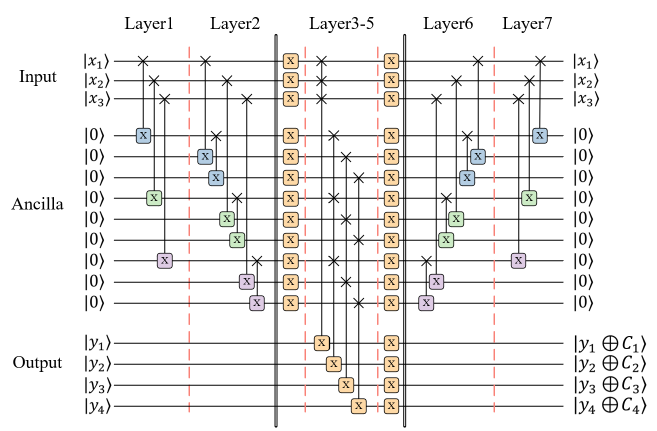}
	\caption{ClausePack oracle for one clustered leaf, instantiated on a maximally shared cluster of four identical clauses \(C_1=\cdots=C_4=(x_1\vee x_2\vee x_3)\).}
	\label{fig:clausepack}
\end{figure}

\subsection{ClausePack Interface}

Consider a feasible clustered leaf \(\mathcal{B}=\{C_{i_1},C_{i_2},\ldots,C_{i_t}\}\) containing \(t\) clauses. Its clause-value vector is \(F_{\mathcal{B}}(x)=\big(C_{i_1}(x),\ldots,C_{i_t}(x)\big)\), where each component is the value of one clause on the assignment \(x\). The input register \(X\) stores \(x\), the temporary register \(W\) stores variable copies used only during the cluster evaluation, and the output register \(Y=(y_1,\ldots,y_t)\) stores the clause values. ClausePack implements the clean reversible map
\begin{equation}
\label{eq:clausepack-spec}
\textsc{ClausePack}_{\mathcal{B}}:\;
|x\rangle_X |0\rangle_W |y\rangle_Y
\mapsto
|x\rangle_X |0\rangle_W
|y\oplus F_{\mathcal{B}}(x)\rangle_Y .
\end{equation}
This clean interface is the property required for a clustered leaf to be composed inside CST-Map.

\subsection{Replication Cost}

ClausePack evaluates all clauses in a cluster in one circuit layer, storing the clause values in distinct ancilla qubits. If several clauses share variable \(z\), one clause-evaluation circuit uses \(|x_z\rangle\) directly, while the other \(\nu_z(\mathcal{B})-1\) circuits use temporary \(|0\rangle\) ancillae prepared by CNOT fan-out from \(|x_z\rangle\); the number of additional variable copies is therefore the redundancy \(R(\mathcal{B})\) in~\eqref{eq:cluster-redundancy}. In an ordered CST partition, the clause-value qubits of earlier clustered leaves remain occupied, whereas the replication region is cleared after each ClausePack evaluation and can be reused by subsequent clustered leaves.

\subsection{Oracle Realization}

ClausePack is implemented by three reversible stages: variable replication, parallel clause evaluation, and reverse replication (Fig.~\ref{fig:clausepack}),
\begin{equation}
\textsc{ClausePack}_{\mathcal{B}}
=
U_{\mathrm{rep}}^\dagger
\,U_{\mathrm{eval}}\,
U_{\mathrm{rep}} .
\end{equation}
The replication stage distributes shared input variables to clean ancilla qubits. For each variable \(z\) with \(\nu_z(\mathcal{B})>1\), a CNOT fan-out tree implements
\begin{equation}
|x_z\rangle |0\rangle^{\otimes(\nu_z(\mathcal{B})-1)}
\mapsto
|x_z\rangle^{\otimes \nu_z(\mathcal{B})},
\qquad x_z\in\{0,1\},
\end{equation}
where the transformation acts on the original input qubit and \(\nu_z(\mathcal{B})-1\) clean ancilla qubits from \(W\). Replication networks for different variables use disjoint temporary qubits and can run in parallel. With \(\nu_{\max}(\mathcal{B})=\max_z \nu_z(\mathcal{B})\le|\mathcal{B}|=t\), the replication depth satisfies
\begin{equation}
D(U_{\mathrm{rep}})\le\lceil \log_2 \nu_{\max}(\mathcal{B})\rceil\le\lceil \log_2 t\rceil .
\end{equation}
This operation is a reversible fan-out of computational-basis values implemented by CNOT gates, a standard logarithmic-depth construction~\cite{ref:moore-nilsson}; it does not clone arbitrary quantum states.

Denote the replicated ancilla state by \(|\chi_{\mathcal{B}}(x)\rangle_W\). The clause circuits are then applied simultaneously:
\begin{equation}
U_{\mathrm{eval}}:\;|x\rangle_X |\chi_{\mathcal{B}}(x)\rangle_W |y\rangle_Y\mapsto|x\rangle_X |\chi_{\mathcal{B}}(x)\rangle_W |y\oplus F_{\mathcal{B}}(x)\rangle_Y .
\end{equation}
If \(D(C_i)\) denotes the depth of the reversible circuit for clause \(C_i\), the parallel evaluation stage has depth
\begin{equation}
D(U_{\mathrm{eval}})=D_{\max}(\mathcal{B})=\max_{C_i\in\mathcal{B}} D(C_i).
\end{equation}

The final stage applies \(U_{\mathrm{rep}}^\dagger\) to erase the replicated variables and restore \(W\) to \(|0\rangle_W\). The composition therefore realizes the clean cluster oracle in \eqref{eq:clausepack-spec}.

\subsection{Depth and Space Cost}

The depth of one ClausePack invocation is bounded by
\begin{equation}
\label{eq:clausepack-depth-bound}
\begin{aligned}
D_{\mathrm{CP}}(\mathcal{B})&\le D_{\max}(\mathcal{B})+2\lceil \log_2 \nu_{\max}(\mathcal{B})\rceil\\
&=O(D_{\max}(\mathcal{B})+\log t).
\end{aligned}
\end{equation}

The ancilla cost consists of the \(t=|\mathcal{B}|\) clause-value ancillae and the \(R(\mathcal{B})\) temporary replication ancillae:
\begin{equation}
\label{eq:clausepack-space-cost}
S_{\mathrm{CP}}(\mathcal{B})
=
|\mathcal{B}|+R(\mathcal{B}).
\end{equation}

Equation~\eqref{eq:clausepack-space-cost} gives the per-cluster ancilla contribution appearing in the CST prefix feasibility constraint in \eqref{eq:cst-prefix-feasibility}.

As a concrete illustration, Fig.~\ref{fig:clausepack} instantiates ClausePack on a maximally shared cluster of four identical clauses \(C_1=\cdots=C_4=(x_1\vee x_2\vee x_3)\). Each of the three variables occurs in all four clauses, so \(\nu_{\max}(\mathcal{B})=4\), \(R(\mathcal{B})=3\times(4-1)=9\), and \(S_{\mathrm{CP}}(\mathcal{B})=4+9=13\). The replication stage (Layers~1--2) fans each variable out to four copies through a depth-\(2\) CNOT tree, the four clauses are then evaluated on disjoint copies in a single parallel layer (Layers~3--5), and the reverse stage (Layers~6--7) restores \(W\) to \(|0\rangle_W\); hence \(D_{\mathrm{CP}}(\mathcal{B})\le D_{\max}(\mathcal{B})+4\), in agreement with~\eqref{eq:clausepack-depth-bound}.

\section{Circuit Mapping from the CST}
\label{sec:mapping}

This section specifies CST-Map, the mapping procedure that takes a feasible CST \(\mathcal{T}=(T,\mathcal{P})\) and synthesizes a reversible quantum circuit implementing the SAT-oracle \(U_f\).

The mapping is defined recursively from the root of \(\mathcal{T}\). Each call to \(\textsc{SynthesizeNode}\) returns both a circuit block and the output qubits exposed to its parent. A clustered leaf \(\mathcal{B}\) is mapped to the ClausePack oracle \(\textsc{ClausePack}_{\mathcal{B}}\) and exposes its clause-value output register. An internal computation node \(v\) first synthesizes all child units and collects their exposed outputs: a clustered child contributes all clause-value qubits produced by ClausePack, whereas a recursive child node contributes its target qubit. An MCX gate controlled by these collected qubits toggles the target qubit \(t_v\) by the conjunction represented by \(v\). For a non-root node, \(t_v\) is a clean ancilla that stores the partial value; for the root, \(t_v\) is the SAT-oracle target qubit. The child circuits are finally applied in reverse order to restore their ancilla qubits.

Fig.~\ref{fig:cst-map-circuit} illustrates this circuit-level realization.

\begin{figure*}[!t]
\centering
\includegraphics[width=\textwidth]{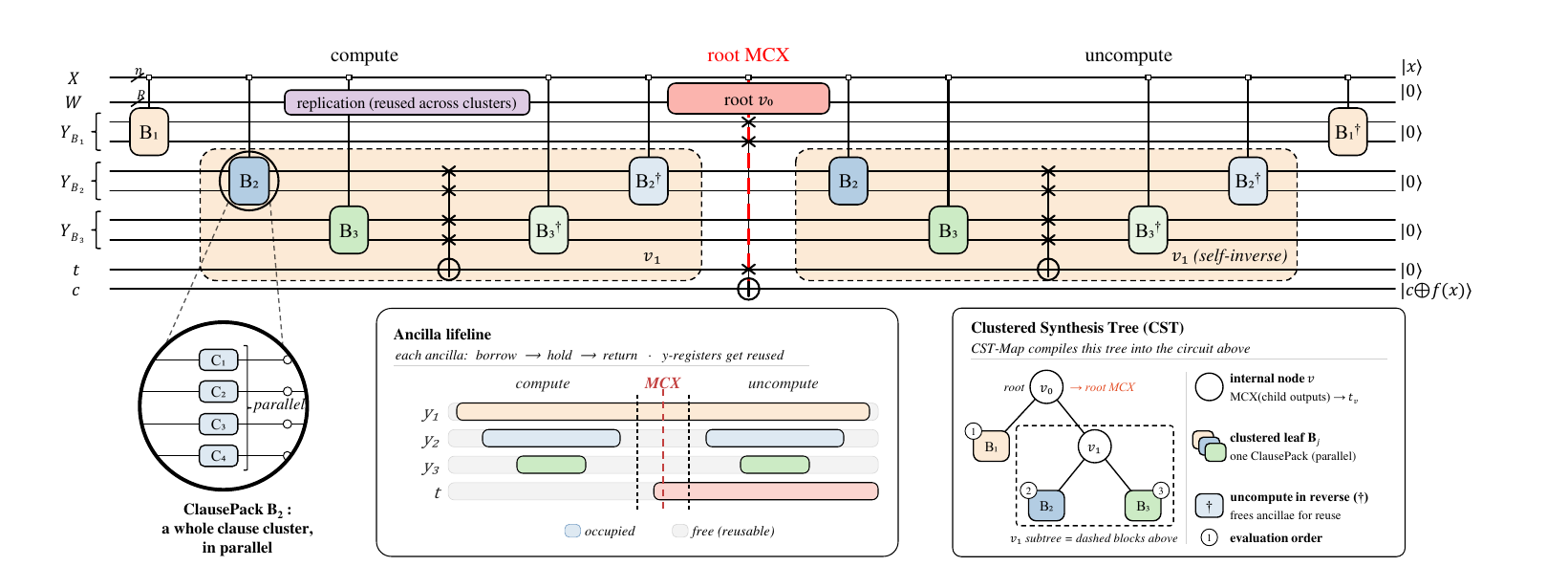}
\caption{Circuit-level illustration of CST-Map. Clustered leaves \(\mathcal{B}_j\) are compiled into \(\textsc{ClausePack}_{\mathcal{B}_j}\) blocks that evaluate clause clusters in parallel; internal CST nodes collect the exposed child outputs, apply an MCX gate to write the node value, and then uncompute the child blocks in reverse order so that temporary ancillae can be reused.}
\label{fig:cst-map-circuit}
\end{figure*}

Algorithm~\ref{alg:cstmap} presents the CST-Map procedure, where \(G_u\) is the circuit block synthesized for unit \(u\), \(O_u\) the output qubits it exposes to its parent, \(Y_{\mathcal{B}}\) the clause-value register of a clustered leaf \(\mathcal{B}\), \(Q_u\) the control qubits collected at an internal node \(u\), \(t_u\) its target qubit, \(\mathcal{G}[w]\) the stored circuit of child \(w\) for the inverse pass, and \(\Vert\) the sequential composition of circuit blocks.

\begin{algorithm}[t]
	\caption{CST-Map: mapping a feasible CST to a SAT-oracle circuit}
	\label{alg:cstmap}
	\begin{algorithmic}[1]
		\Require feasible CST \(\mathcal{T}=(T,\mathcal{P})\)
		\Ensure reversible quantum circuit implementing \(U_f\)
		\State \((G_{\mathrm{root}},O_{\mathrm{root}})\gets\textsc{SynthesizeNode}(\textsc{Root}(T),\,\mathcal{T})\)
		\State \Return \(G_{\mathrm{root}}\)
		\Statex
		\Function{SynthesizeNode}{$u,\mathcal{T}$}
			\If{\(u\) is a clustered leaf \(\mathcal{B}\)}
				\State \(Y_{\mathcal{B}}\gets\textsc{ClauseValueRegister}(\mathcal{B})\)
				\State \Return \((\textsc{ClausePack}_{\mathcal{B}},\,Y_{\mathcal{B}})\)
			\EndIf
			\State \(G_u\gets\emptyset\)
			\State \(Q_u\gets\emptyset\)
			\State \(\mathcal{G}\gets\emptyset\)
			\ForAll{\(w\in\textsc{Children}_{\mathcal{T}}(u)\) in evaluation order}
				\State \((G_w,O_w)\gets\textsc{SynthesizeNode}(w,\,\mathcal{T})\)
				\State \(\mathcal{G}[w]\gets G_w\)
				\State \(G_u\gets G_u\,\Vert\,G_w\) \Comment{evaluate child result}
				\State \(Q_u\gets Q_u\cup O_w\)
			\EndFor
			\State \(G_u\gets G_u\,\Vert\,\textsc{MCX}(Q_u;t_u)\) \Comment{write the value of \(u\)}
			\ForAll{\(w\in\textsc{Children}_{\mathcal{T}}(u)\) in reverse evaluation order}
				\State \(G_u\gets G_u\,\Vert\,\mathcal{G}[w]^{\dagger}\) \Comment{restore child registers}
			\EndFor
			\State \Return \((G_u,\{t_u\})\)
		\EndFunction
	\end{algorithmic}
\end{algorithm}

Each call to \(\textsc{SynthesizeNode}(u)\) computes the value of the subformula rooted at \(u\), exposes that value to its parent, and clears the temporary qubits it used. A clustered leaf has this property by the clean map~\eqref{eq:clausepack-spec}, and an internal node by uncomputing its children after the MCX gate. Therefore, a non-root call leaves only \(t_u\) as its output, while the root call toggles the SAT-oracle target by \(f(x)\), realizing \(U_f\).

The resource usage of the mapped circuit follows from the CST and its clustered leaves: the prefix feasibility constraint~\eqref{eq:cst-prefix-feasibility} keeps the clause-value and replication ancillae of each cluster within the budget, and the ClausePack bounds~\eqref{eq:clausepack-depth-bound} and~\eqref{eq:clausepack-space-cost} give the depth and ancilla contribution of each clustered leaf. Since CST-Map only traverses the given CST, its mapping time is linear in the size of the synthesized circuit.

\section{Experiments}
\label{sec:experiments}

This section reports our experimental results. First, we compare the depth of the oracles synthesized by CST with those of the baselines under the same ancilla budgets, on SATLIB and under the configurations used in the baseline papers. Second, we isolate where the reduction comes from, separating the contributions of ClausePack and SeedGrow and varying the instance structure and the ancilla budget. Finally, we estimate the cost of Grover search built on the resulting oracles.

\subsection{Experimental Setup}
\label{subsec:setup}

We use the SATLIB benchmark suite~\cite{hoos2000satlib} as the main experimental benchmark to test the performance of different algorithms. It provides random $3$-CNF families and structured families whose clause widths are mixed. We also use the random $4$-CNF family of~\cite{yang2024sat}, so that the pAND constructions are also compared on the benchmark reported in their own paper. To measure the effect of clause grouping and variable sharing, we generate hot-spot instances, in which the literals of a clause are drawn from a small pool of hot variables, together with graph-coloring and pigeonhole formulas. For families with repeated instances we use the average depth to measure the performance of the algorithms.

We compare CST with four baselines: the depth-oriented and the width-oriented pAND constructions~\cite{yang2024sat}, clause storing (CS), and qSAT~\cite{kole2025qsat,assaad2026esop}. The pAND constructions synthesize a SAT-oracle under an adjustable ancilla budget. CS evaluates each clause on its own work qubit, combines the results, and uncomputes. qSAT uses an ESOP-based encoding and fixes its own ancilla count. The pAND variants and qSAT are run with the authors' implementations, and CS is an analytic model transcribed from Qiskit Aqua 0.9.5~\cite{ref:qiskit-aqua}.\footnote{For qSAT we report the depth returned by its toolchain, as in~\cite{assaad2026esop}. This value is the largest number of gates acting on any single qubit of the Clifford+T circuit, which is a lower bound on the circuit depth, so the reported reductions relative to qSAT are conservative.}

The number of ancilla qubits $a_q$ (i.e., qubits that do not hold the inputs or the oracle target) and the logical depth of the synthesized circuit over the Clifford+T gate set~\cite{amy2013} are used as metrics to quantify the cost of the final circuit, so the circuit width is $n+a_q+1$. The depth is computed from the analytic bounds of Sections~\ref{sec:method} and~\ref{sec:mapping}. By running our algorithm on random $3$-CNF instances of $n=6,\dots,20$ variables, the experiment results show that the depth measured on the transpiled circuit aligns well with our theoretical analysis, which never falls below it and exceeds it by $3.9\%$ on average. In the depth ratio columns, we present the depth of the baseline divided by the depth of CST, so that values above one indicate a reduction.

\subsection{Overall Comparison}
\label{subsec:overall}

To evaluate the effectiveness of CST in reducing the circuit depth under limited ancillae, we conduct comparative experiments on the SATLIB benchmarks and the random $4$-CNF benchmark of~\cite{yang2024sat}. To compare the constructions under the same ancilla limits, we sweep the ancilla budget $a_q$ up to $2m-1$, over $18$ evenly spaced points from $n$ on SATLIB and over a grid extended downward to $a_q=8$ on the random $4$-CNF benchmark, and report the circuit depth at each budget as the primary evaluation metric. To further compare CST with CS and qSAT, which fix their own ancilla counts, we run CST at $a_q=2m-1$ for CS and at the ancilla count $a_q^{\mathrm{qSAT}}$ of qSAT, or at the nearest sampled budget on the random $4$-CNF benchmark. The results are presented in Figs.~\ref{fig:satlib} and~\ref{fig:e2e-sota-full} and Table~\ref{tab:satlib}.

\begin{table*}[t]
	\centering
	\caption{Comparison of circuit depth between CST and the baselines on the SATLIB benchmarks}
	\label{tab:satlib}
	\scriptsize
	\setlength{\tabcolsep}{3.5pt}
	\renewcommand{\arraystretch}{1.12}
	\begin{tabular}{lcccc cccc}
	\toprule
	\multirow{2}{*}{Family} & \multicolumn{4}{c}{Benchmark characteristics} & \multicolumn{4}{c}{CST same-budget speedup $D_{\text{baseline}}/D_{\text{CST}}$} \\
	\cmidrule(lr){2-5}\cmidrule(lr){6-9}
	 & $n$ & $m$ & $k$ & \#inst. & \makecell{vs.\ depth-oriented pAND\\($a_q{\approx}n \to 2m{-}1$)} & \makecell{vs.\ width-oriented pAND\\($a_q{\approx}n \to 2m{-}1$)} & \makecell{vs.\ CS\\($a_q{=}2m{-}1$)} & \makecell{vs.\ qSAT\\($a_q{=}a_q^{\text{qSAT}}$)} \\
	\midrule
	\multicolumn{9}{l}{\emph{Random group (random 3-CNF)}} \\
	uf50-218 & 50 & 218 & 3 & 100 & $3.0\times\to6.9\times$ & $9.4\times\to11.9\times$ & $19.2\times$ & $8.5\times$ \\
	uf100-430 & 100 & 430 & 3 & 100 & $2.8\times\to6.9\times$ & $16.0\times\to21.4\times$ & $34.3\times$ & $15.5\times$ \\
	uf150-645 & 150 & 645 & 3 & 100 & $2.7\times\to6.6\times$ & $22.0\times\to29.7\times$ & $47.4\times$ & $21.6\times$ \\
	uf200-860 & 200 & 860 & 3 & 100 & $2.6\times\to6.5\times$ & $28.1\times\to38.2\times$ & $60.8\times$ & $27.9\times$ \\
	uf250-1065 & 250 & 1065 & 3 & 100 & $2.6\times\to6.1\times$ & $34.1\times\to44.2\times$ & $70.2\times$ & $32.4\times$ \\
	\midrule
	\multicolumn{9}{l}{\emph{Structured group (mixed clause width)}} \\
	flat30-60 & 90 & 300 & 3 & 100 & $3.2\times\to4.3\times$ & $14.4\times\to13.4\times$ & $3.8\times$ & $1.3\times$ \\
	flat50-115 & 150 & 545 & 3 & 100 & $3.4\times\to4.7\times$ & $22.4\times\to22.1\times$ & $5.3\times$ & $2.0\times$ \\
	flat100-239 & 300 & 1117 & 3 & 100 & $3.1\times\to4.3\times$ & $38.9\times\to40.0\times$ & $8.5\times$ & $3.5\times$ \\
	ais8 & 113 & 1520 & 8 & 1 & $10.2\times\to19.6\times$ & $16.1\times\to50.2\times$ & $29.3\times$ & $12.6\times$ \\
	ais10 & 181 & 3151 & 10 & 1 & $8.6\times\to26.1\times$ & $22.1\times\to89.2\times$ & $49.7\times$ & $22.2\times$ \\
	ais12 & 265 & 5666 & 12 & 1 & $9.7\times\to33.7\times$ & $29.9\times\to144\times$ & $79.2\times$ & -- \\
	bw\_large.a & 459 & 4675 & 10 & 1 & $6.9\times\to43.2\times$ & $53.2\times\to120\times$ & $73.7\times$ & -- \\
	bw\_large.b & 1087 & 13772 & 12 & 1 & $7.5\times\to56.7\times$ & $102\times\to320\times$ & $194\times$ & -- \\
	\bottomrule
	\end{tabular}
\end{table*}

\begin{figure}[t]
	\centering
	\includegraphics[width=\columnwidth]{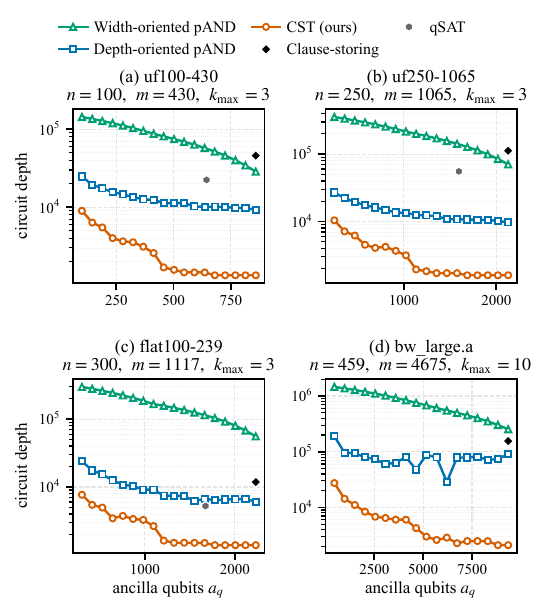}
	\caption{Circuit depth versus ancilla budget. (a) \texttt{uf100}. (b) \texttt{uf250}. (c) \texttt{flat100}. (d) \texttt{bw\_large.a}.}
	\label{fig:satlib}
\end{figure}

As shown in Fig.~\ref{fig:satlib} and Table~\ref{tab:satlib}, CST consistently achieves the lowest circuit depth among all constructions on all $13$ SATLIB families and at every sampled budget. Averaged over the $13$ families and $18$ sampled budgets, CST achieves a speedup of $10.2\times$ over the depth-oriented pAND construction, ranging from $2.6\times$ on \texttt{uf250-1065} at $a_q\approx n$ to $56.7\times$ on \texttt{bw\_large.b} at $a_q=2m-1$. Notably, on the \texttt{flat} families, whose clauses mostly carry two literals, CST still achieves $3.1\times$ to $4.7\times$ over the depth-oriented pAND construction. Even against CS and qSAT at their own ancilla counts, CST remains shallower, by $3.8\times$ to $194\times$ and by $1.3\times$ to $32.4\times$, respectively.

\begin{figure}[t]
	\centering
	\includegraphics[width=\columnwidth]{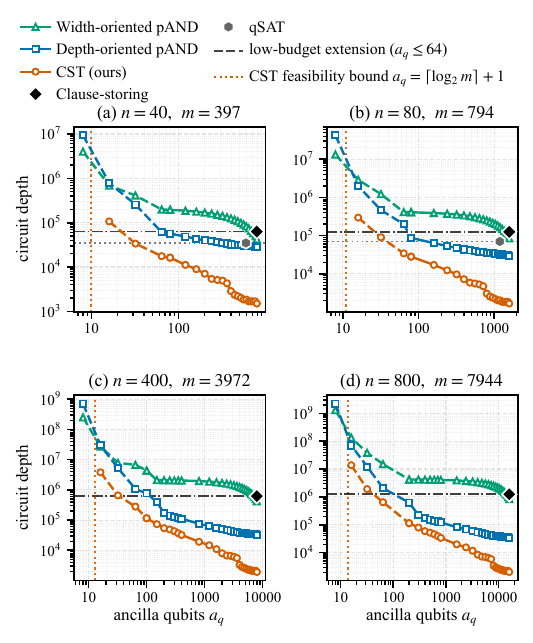}
	\caption{Circuit depth versus ancilla budget $a_q$ on the random $4$-CNF size classes (mean over 20 instances; logarithmic axes). Solid: main grid; dashed: low-budget extension $a_q\in\{8,16,32,64\}$; vertical dotted line: CST feasibility bound $a_q=\lceil\log_2 m\rceil+1$. Horizontal lines extend the CS value at $a_q=2m-1$ and the qSAT value at $a_q^{\text{qSAT}}$ across the panel. (a) $n=40$. (b) $n=80$. (c) $n=400$. (d) $n=800$.}
	\label{fig:e2e-sota-full}
\end{figure}

On the random $4$-CNF benchmark of~\cite{yang2024sat} (Fig.~\ref{fig:e2e-sota-full}), CST reduces the circuit depth, achieving an average speedup of $8.8\times$ to $11.0\times$ over the depth-oriented pAND construction across the four size classes. The budgets of a few tens to a few hundred ancillae present the worst result, where the speedup falls to $2.8\times$ to $3.5\times$. However, at $a_q\approx m$ and above, where every clause can hold its own ancilla, the speedup exceeds $11\times$ and reaches $16.6\times$ to $18.7\times$ at $2m-1$. In comparison, the width-oriented pAND construction is $31$ to $371$ times deeper than CST on average, CS at $2m-1$ is $42$ to $613$ times deeper, qSAT at $a_q^{\mathrm{qSAT}}=594$ and $1{,}190$ is $18.9\times$ and $34.8\times$ deeper, and qSAT is not applicable to the size classes $n=400$ and $n=800$, terminating with \texttt{\upshape std::bad\_alloc}. On the other hand, since the pAND constructions divide the CNF into blocks recursively, they can be built with fewer ancillae than CST, down to $a_q=8$, and $a_q=4$ for the width-oriented one, whereas CST requires at least $\lceil\log_2 m\rceil+1$ ancillae (Section~\ref{subsec:backbone}).

On both benchmarks, the advantages of CST stem from its clustered leaves, which evaluate the clauses of each cluster in one circuit layer through ClausePack to exploit clause-level parallelism within the ancilla budget. The same trend is observed across varying benchmark families, clause widths, and ancilla budgets.

\subsection{Ablation Study}
\label{subsec:grouping}

We conduct ablation experiments on ClausePack, SeedGrow, and the sensitivity of SeedGrow to the clause width, the hot-variable ratio, the ancilla budget, and the instance family. To verify the validity of key components, we use hot-spot $4$-CNF instances as the experimental benchmark, where the number of variables $n$ is $40$, the clause count follows the phase transition of random $4$-SAT, $m=\lfloor 9.931\,n\rfloor=397$, and the ancilla budget $a_q$ is $40$, and subsequent settings are the same if not specifically stated.

\textit{ClausePack and SeedGrow:} We compare the No-grouping backbone with the six grouping rules and SeedGrow in Table~\ref{tab:grouping}, with each rule extended to the prefix-dependent capacity~\eqref{eq:cst-prefix-feasibility} and the depths normalized to the No-grouping backbone. The depth reduction originates primarily from clustering the leaves: Random reaches $0.245$ and Next-Fit $0.243$, so ClausePack alone removes about three quarters of the depth. The grouping rule determines the number of clusters, which is the number of serial stages at a node, and this number becomes harder to reduce as the ancilla budget tightens. Since the problem of minimizing it is NP-hard (Theorem~\ref{thm:accs-npc}), SeedGrow is used at this step. Notably, SeedGrow attains the lowest depth of $0.201$, $17.5\%$ to $18.0\%$ below the two sequential packers and $43.5\%$ to $44.2\%$ below the four coloring rules. Moreover, on $140$ small instance--budget pairs ($n=20$, $m\in\{8,10,12\}$) with the exhaustively enumerated optimum, the gap of SeedGrow averages $2.6\%$, compared with $8.9\%$ for Next-Fit.

\begin{table}[t]
	\centering
	\caption{Clause-grouping rules on the hot-spot grid and on the exact-optimum set. Depths are normalized to the No-grouping backbone (lower is better). Depth@ref: reference configuration; Avg.\ depth: mean over the fixed-budget sweeps of Fig.~\ref{fig:grouping-axes}. Gap: relative depth excess over the exact optimum; Hits opt.: cells attaining it. ``--'': exact optimum not available at $n=40$. Unless otherwise specified, the instances are produced by a hot-spot generator that draws each literal from a pool of $H=8$ hot variables with probability $\rho=0.4$ and uniformly from all $n$ variables otherwise, and each entry is the mean over $20$ random instances.}
	\label{tab:grouping}
	\scriptsize
	\setlength{\tabcolsep}{2pt}
	\renewcommand{\arraystretch}{1.15}
	\begin{tabular}{@{}lcc@{\hspace{4pt}}ccc@{}}
	\toprule
	& \multicolumn{2}{c}{Hot-spot grid ($n=40$)} & \multicolumn{3}{c}{Exact-optimum set ($n{=}20$)} \\
	\cmidrule(lr){2-3}\cmidrule(lr){4-6}
	Method & Depth@ref & Avg.\ depth & \makecell{Gap avg.\\(\%)} & \makecell{Gap max.\\(\%)} & \makecell{Hits\\opt.} \\
	\midrule
	No-grouping (ShallowGrow) & 1.000 & 1.000 & 200.8 & 259.3 & 0/140 \\
	\midrule
	Random & 0.245 & 0.253 & 9.1 & 38.1 & 2/140 \\
	Next-Fit~\cite{johnson1974fast} & 0.243 & 0.255 & 8.9 & 45.4 & 22/140 \\
	Welsh-Powell~\cite{welsh1967upper} & 0.360 & 0.385 & 69.2 & 127.5 & 1/140 \\
	Matula-Beck~\cite{matula1983smallest} & 0.358 & 0.383 & 69.3 & 127.5 & 1/140 \\
	DSATUR~\cite{brelaz1979new} & 0.356 & 0.377 & 69.2 & 127.5 & 1/140 \\
	RLF~\cite{leighton1979graph} & 0.355 & 0.376 & 67.5 & 127.5 & 1/140 \\
	\textbf{SeedGrow (ours)} & \textbf{0.201} & \textbf{0.213} & \textbf{2.6} & \textbf{28.5} & \textbf{41/140} \\
	\midrule
	Exact optimum & -- & -- & 0.0 & 0.0 & 140/140 \\
	\bottomrule
	\end{tabular}
\end{table}

\textit{Sensitivity to Instance Parameters and Ancilla Budget:} We analyze the sensitivity of SeedGrow and the six comparisons to the clause width $k$, the hot-variable ratio $\rho$, the ancilla budget $a_q$, and the instance family. The clause width is swept over $k\in\{3,4,5,6\}$, the hot-variable ratio over $\rho\in\{0,0.2,0.4,0.6,0.8\}$ with a hot pool of $H=8$, and the ancilla budget over nine values from $40$ to $793$. The two structured families are graph coloring ($n=36$, $m=78$--$138$, $k=3$) and pigeonhole ($n=26$--$54$, $m=58$--$176$, $k=4$--$8$), each on $20$ instances. The ratio of the depth of each rule to that of SeedGrow is referred to as the relative depth, and Figs.~\ref{fig:grouping-axes}(a)--(d) show it under different clause widths, hot-variable ratios, ancilla budgets, and instance families, respectively.

\begin{figure}[t]
	\centering
	\includegraphics[width=\columnwidth]{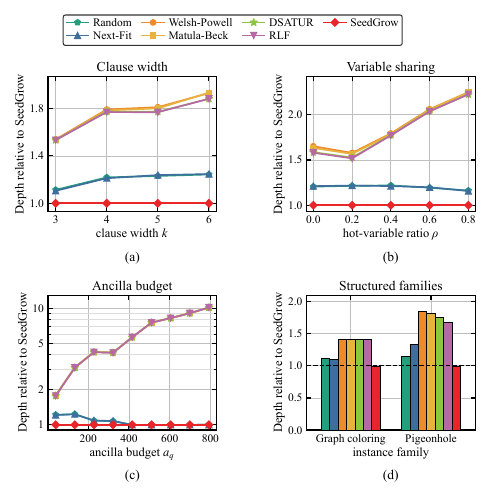}
	\caption{Depth of the grouping rules relative to SeedGrow versus (a) clause width $k$, (b) hot-variable ratio $\rho$, (c) ancilla budget $a_q$, and (d) structured families. SeedGrow is the reference at $1.0$; lower is better.}
	\label{fig:grouping-axes}
\end{figure}

Fig.~\ref{fig:grouping-axes}(a) shows that the relative depth of Next-Fit increases from $1.11$ at $k=3$ to $1.25$ at $k=6$. This is because the evaluation of the clustered leaves is the only part of the depth that depends on the grouping, and its share of the depth of SeedGrow increases from $43.1\%$ to $82.9\%$, as $m$ grows from $170$ to $1{,}734$ under the same ancilla budget. Fig.~\ref{fig:grouping-axes}(b) shows that the relative depth of Next-Fit ranges from $1.20$ to $1.22$ for $\rho\le0.6$ and decreases to $1.16$ at $\rho=0.8$. This is because almost every clause contains a hot variable at $\rho=0.8$, and thus each clause added to a cluster duplicates $2.12$ variables on average under SeedGrow and $2.23$ under Next-Fit, compared with $1.21$ and $1.48$ at $\rho=0.4$.

Fig.~\ref{fig:grouping-axes}(c) shows that as $a_q$ increases, the relative depth of Next-Fit decreases from $1.21$ at $a_q=40$ to $1.00$ at $a_q=416$, while that of the best coloring rule increases from $1.77$ to $10.19$ at $a_q=793$. This is because once $a_q$ exceeds $m=397$, all the clauses are grouped at a single node of the synthesis tree, where Next-Fit and SeedGrow form the same average number of clusters, whereas the number of color classes does not depend on $a_q$. Fig.~\ref{fig:grouping-axes}(d) shows that the relative depth of Next-Fit is $1.10$ on the graph-coloring family and $1.34$ on the pigeonhole family. This is because Next-Fit packs the clauses in the input order, and $69.3\%$ of the consecutive clause pairs in the pigeonhole instances share a variable, compared with none in the graph-coloring instances. The results indicate that at $a_q=40$, SeedGrow is shallower than all the compared rules for every clause width, hot-variable ratio, and instance family, whereas once $a_q$ exceeds $m$, the relative depths of the two sequential packers range from $0.997$ to $1.001$.

\subsection{Resource Estimation for Grover Search}
\label{subsec:grover-resource}

The SAT problem is to determine whether a Boolean function $f$ in CNF over $n$ variables is satisfiable, i.e., whether there exists an assignment $x\in\{0,1\}^{n}$ such that $f(x)=1$. Grover's algorithm is the most direct quantum approach to this problem: it performs an unstructured search over the $2^{n}$ assignments and queries an oracle that evaluates $f$ only $O(\sqrt{2^{n}})$ times~\cite{ref:grover-oracle,ref:sat-in-quantum}.

We apply CST to estimate the quantum resources required for solving random $k$-SAT with Grover's algorithm, where CST synthesizes the oracle called in each Grover iteration. Following~\cite{yang2024sat}, the number of ancillary qubits is fixed at $240$, the clause width is $k=3$, $5$, and $7$, and the number of variables is $n=40$ and $80$. For each $k$, the clause count is set at the satisfiability threshold, i.e., $m=\lfloor 4.267\,n\rfloor$, $\lfloor 21.117\,n\rfloor$, and $\lfloor 87.79\,n\rfloor$, respectively, and $20$ instances are generated for each pair of $k$ and $n$.

In Grover's algorithm, each of the $R_{\mathrm{G}}=\lfloor(\pi/4)\sqrt{2^{n}}\rfloor$ Grover iterations applies the oracle followed by the diffusion operator. Table~\ref{tab:grover-resource} lists the circuit width, the depth of one round, and the depth of the full $R_{\mathrm{G}}$ rounds, where the values of the depth-oriented pAND construction are taken from~\cite{yang2024sat}. For instance, random $7$-SAT with $n=80$ uses $321$ qubits, i.e., $80$ variable qubits, $240$ ancillary qubits, and one output qubit, and the full Grover circuit built on CST has an average depth of $2.72\times10^{17}$, compared with $9.3\times10^{17}$ for the depth-oriented pAND construction.

\begin{table}[t]
	\centering
	\caption{Resource estimation for Grover search on random $k$-SAT with the depth-oriented pAND oracle~\cite{yang2024sat} and the CST oracle at $a_q=240$ (pAND: values reported in Table~I of~\cite{yang2024sat}; CST: mean over 20 instances). The number of Grover rounds is $R_{\mathrm{G}}=8.24\times10^{5}$ for $n=40$ and $8.64\times10^{11}$ for $n=80$, and the circuit width $n+a_q+1$ is $281$ and $321$, respectively. One round consists of one oracle call and one diffusion operator. Norm.: CST one-round depth normalized to pAND.}
	\label{tab:grover-resource}
	\footnotesize
	\setlength{\tabcolsep}{3pt}
	\renewcommand{\arraystretch}{1.12}
	\begin{tabular}{ccccccc}
		\toprule
		& & \multicolumn{3}{c}{One-round depth} & \multicolumn{2}{c}{Full-round depth} \\
		\cmidrule(lr){3-5}\cmidrule(lr){6-7}
		$n$ & $m$ & pAND & CST & Norm. & pAND & CST \\
		\midrule
		\multicolumn{7}{l}{$k=3$ \quad ($m=\lfloor 4.267\,n\rfloor$)} \\
		40 & 170  & 7{,}523           & 1{,}846   & 0.245 & $6.2\times10^{9}$  & $1.52\times10^{9}$ \\
		80 & 341  & 11{,}586          & 4{,}152   & 0.358 & $1.0\times10^{16}$ & $3.59\times10^{15}$ \\
		\midrule
		\multicolumn{7}{l}{$k=5$ \quad ($m=\lfloor 21.117\,n\rfloor$)} \\
		40 & 844  & 92{,}444          & 24{,}773  & 0.268 & $7.6\times10^{10}$ & $2.04\times10^{10}$ \\
		80 & 1689 & $1.3\times10^{5}$ & 43{,}670  & 0.336 & $1.1\times10^{17}$ & $3.77\times10^{16}$ \\
		\midrule
		\multicolumn{7}{l}{$k=7$ \quad ($m=\lfloor 87.79\,n\rfloor$)} \\
		40 & 3511 & $8.1\times10^{5}$ & 200{,}755 & 0.248 & $6.7\times10^{11}$ & $1.65\times10^{11}$ \\
		80 & 7023 & $1.1\times10^{6}$ & 314{,}674 & 0.286 & $9.3\times10^{17}$ & $2.72\times10^{17}$ \\
		\bottomrule
	\end{tabular}
\end{table}

As shown in Table~\ref{tab:grover-resource}, both the one-round and the full-round depths of CST are $0.25$--$0.36\times$ those of the depth-oriented pAND construction. This ratio increases with $n$ for every $k$: from $0.25$, $0.27$, and $0.25$ at $n=40$ to $0.36$, $0.34$, and $0.29$ at $n=80$ for $k=3$, $5$, and $7$, respectively. This is because $a_q$ is fixed at $240$ while $m$ doubles from $n=40$ to $n=80$, so the budget admits only smaller clusters, which leaves less room for improvement, as observed in Section~\ref{subsec:overall}. Note that these depths do not include the overhead of quantum error correction. These results demonstrate that the depth reduction of the CST oracle carries over to the full Grover search under the ancilla budget of~\cite{yang2024sat}.

\section{Conclusion}
\label{sec:conclusion}

This article studied the synthesis of depth-efficient SAT-oracles under a limited ancilla budget. It first illustrated the importance of parallel clause evaluation for the oracle depth under such a budget, where insufficient ancilla space forces many clause evaluations to be serialized or recomputed. Then, the \emph{Clustered Synthesis Tree} (CST) was proposed to reduce the circuit depth while preserving the given ancilla budget, by replacing the singleton clause leaves of the FORM/ShallowGrow backbone with clustered leaves. On top of this backbone, the synthesis was divided into three phases, including clause grouping, parallel clause evaluation, and circuit mapping. First, SeedGrow was developed to construct feasible cluster partitions under the ancilla budget in \(O(m^2k)\) time and \(O(mk)\) space, where \(m\) is the number of clauses and \(k\) is the maximum clause width. Then, each cluster was realized by ClausePack, which evaluates its clauses in parallel, handles shared variables through a controlled replication cost, and attains a depth of \(O(D_{\max}(\mathcal{B})+\log t)\) for a cluster of \(t\) clauses. Finally, the feasible CST was compiled by CST-Map into an executable SAT-oracle circuit within the given ancilla budget.

Evaluation results demonstrate significant depth reductions under the same ancilla budgets, with average reductions of \(10.2\times\) and \(8.8\times\)--\(11.0\times\) over the depth-oriented pAND baseline on SATLIB and random \(4\)-CNF, respectively, and CST also outperforms all the other baselines on both benchmarks. The contributions of ClausePack and SeedGrow are also validated through comparative experiments, where ClausePack alone removes about three quarters of the backbone depth and SeedGrow further reduces it by addressing the NP-hard problem of minimizing the number of clusters, achieving the lowest depth among the compared grouping rules and near-optimal results on small instances. The resource estimation for Grover search further highlights the potential of CST in quantum search, where the Grover circuit built on the CST oracle requires only \(25\%\)--\(36\%\) of the depth of that built on the depth-oriented pAND oracle under the same ancilla budget.

For future work, we plan to develop more effective grouping algorithms and more refined redundancy models for variables shared across adjacent clusters. Furthermore, we intend to extend CST to account for hardware constraints, such as limited qubit connectivity and architecture-dependent gate costs.

\end{document}